\documentclass[11pt]{article}

\usepackage{amsmath,amssymb,amsthm,amsfonts,latexsym,bbm,xspace,graphicx,float,mathtools,braket,bbold}
\usepackage[letterpaper,margin=1in]{geometry}
\usepackage[pagebackref,colorlinks,citecolor=blue,linkcolor=blue,bookmarks=true]{hyperref}
\usepackage{enumitem}
\usepackage[nameinlink,capitalize]{cleveref}
\usepackage[mode=buildnew]{standalone}

\usepackage{graphicx}
\usepackage[many]{tcolorbox}
\newtcolorbox[
  auto counter,
  crefname={Algorithm}{Algorithms}]%
{myalgorithm}[2][]{
breakable,
enhanced,
  fonttitle=\bfseries,
  code={\def\mytitle{#2}},
  title=Algorithm \thetcbcounter: #2,%
  #1}

\usepackage{sidecap}
\sidecaptionvpos{figure}{c}

\usepackage[colorinlistoftodos]{todonotes}
\usepackage{algorithm}
\usepackage[noend]{algpseudocode}
\usepackage{tikz-3dplot}

\renewcommand{\backref}[1]{}

\renewcommand{\backrefalt}[4]{%
\ifcase #1 %
\or
[p.\ #2]%
\else
[pp.\ #2]%
\fi}

\newtheorem{theorem}{Theorem}[section]
\newtheorem*{namedtheorem}{\theoremname}
\newcommand{\theoremname}{testing}

\newtheorem{lemma}[theorem]{Lemma}

\theoremstyle{definition}

\renewcommand{\Pr}{\mathop{\bf Pr\/}}

\newcommand{\tr}{\mathrm{tr}}

\newcommand{\R}{\mathbb R}
\newcommand{\C}{\mathbb C}

\newcommand{\eps}{\epsilon}

\newcommand{\calD}{\mathcal{D}}

\newcommand{\abs}[1]{\lvert #1 \rvert}
\newcommand{\Abs}[1]{\Bigl\lvert #1 \Bigr\rvert}

\newcommand{\norm}[1]{\lVert #1 \rVert}

\newcommand{\ignore}[1]{}

\newcounter{termcounter}[equation]
\renewcommand{\thetermcounter}{\the\numexpr\value{equation}+1\relax.\roman{termcounter}}
\crefname{term}{term}{terms}
\creflabelformat{term}{#2\textup{(#1)}#3}

\makeatletter
\def\term{\@ifnextchar[\term@optarg\term@noarg}
\def\term@optarg[#1]#2{%
  \textup{#1}%
  \def\@currentlabel{#1}%
  \def\cref@currentlabel{[][2147483647][]#1}%
  \cref@label[term]{#2}}
\def\term@noarg#1{%
  \refstepcounter{termcounter}%
  \textup{\thetermcounter}%
  \cref@label[term]{#1}}
\makeatother

\title{Efficient Tomography of Non-Interacting Fermion States}

\author{
Scott Aaronson\thanks{\href{mailto:aaronson@cs.utexas.edu}{\texttt{aaronson@cs.utexas.edu}}. \ Supported by a Vannevar Bush Fellowship from the US Department of Defense, the Berkeley NSF-QLCI CIQC Center, a Simons Investigator Award, and the Simons “It from Qubit” collaboration.}\\
\small{\sl University of Texas at Austin }
\and Sabee Grewal\thanks{\href{mailto:sabee@cs.utexas.edu}{\texttt{sabee@cs.utexas.edu}}. \ Supported by Scott Aaronson's Vannevar Bush Fellowship from the US Department of Defense, the Berkeley NSF-QLCI CIQC Center, a Simons Investigator Award, and the Simons “It from Qubit” collaboration.}\\
\small{\sl University of Texas at Austin}
}

\date{}

\begin{document}

\maketitle

\begin{abstract}
We give an efficient algorithm that learns a non-interacting fermion state, given copies of the state. \
For a system of $n$ non-interacting fermions and $m$ modes, we show that $O(m^3 n^2 \log(1/\delta) / \eps^4)$ copies of the input state and $O(m^4 n^2 \log(1/\delta)/ \eps^4)$ time are sufficient to learn the state to trace distance at most $\eps$ with probability at least $1 - \delta$. \
Our algorithm empirically estimates one-mode correlations in $O(m)$ different measurement bases and uses them to reconstruct a succinct description of the entire state efficiently. 
\end{abstract}

\section{Introduction}
There are two types of particles in the universe: bosons and fermions. \ 
Bosons include force carriers, such as photons and gluons, and fermions include matter particles like quarks and electrons. \  
Each particle can be in a certain mode (e.g., a position or state). \
For a system of $n$ particles, a configuration of the system is described by specifying how many particles are in each of $m$ modes. \
Bosons are particles where multiple occupancy of a mode is allowed, whereas fermions are particles where multiple occupancy is forbidden;
that is, two or more fermions cannot occupy the same mode at once (this is the \textit{Pauli exclusion principle}). \
It follows that a system of $n$ fermions and $m$ modes has $\binom{m}{n}$ possible configurations; we denote the set of possible configurations by $\Lambda_{m,n}$.

Our main result is an efficient algorithm (both in copy complexity and time complexity) for learning a non-interacting fermion state (also called a free-fermion state or a Gaussian fermion state), which is a superposition over configurations in $\Lambda_{m,n}$. \ 
Even though a non-interacting fermion state lives in an exponentially large Hilbert space, we show how to exploit its structure to output a succinct description of the state efficiently. \
A non-interacting fermion state can be completely specified by an $m \times n$ column-orthonormal matrix $A$. \
Our algorithm measures copies of the input state in $O(m)$ different measurement bases, and uses the measurement data to reconstruct an $m \times n$ matrix $\hat{A}$ in polynomial time. \
We prove that a polynomial number of copies of the input state is enough for the output state to be $\eps$-close to the original state in trace distance. 

\begin{theorem}[Main result]
Let $\ket{\Psi}$ be a state of $n$ non-interacting fermions and $m$ modes. \ 
There exists an algorithm that uses $O(m^3 n^2 \log(1/\delta)/ \eps^4)$ copies of $\ket{\Psi}$, $O(m^4 n^2 \log(1/\delta)/\eps^4)$ classical time, and $O(m)$ measurement bases, and outputs a succinct description of a non-interacting fermion state $\ket{\hat{\Psi}}$ that is $\eps$-close in total variation distance to $\ket{\Psi}$ with probability at least $1 - \delta$.
\end{theorem}

Our algorithm can also be adapted to conventional quantum state tomography, which we explain in \cref{sec:tomography}. \ 

\subsection{Main Ideas}\label{subsec:main-ideas}
Here and throughout, let $U \in \C^{m \times m}$ be the unitary that prepares the unknown non-interacting fermion state $\ket{\Psi}$ from the standard initial state $\ket{1_n}$ (the state where the first $n$ modes are occupied and the remaining are unoccupied), and let $A \in \C^{m \times n}$ be the column-orthonormal matrix corresponding to the first $n$ columns of $U$. \
Define $K = (k_{ij}) \coloneqq A A^\dagger \in \C^{m \times m}$. \ We refer to $K$ as the \textit{kernel matrix} due to the connection between determinantal point processes and non-interacting fermions (which we discuss further in \cref{subsec:dpp}). \ In the physics literature, the kernel matrix is also called the one-body reduced density matrix (1-RDM) or the correlation matrix. \

The elements of $\Lambda_{m,n}$ are the possible configurations of a system of $n$ non-interacting fermions and $m$ modes. \ Formally, $\Lambda_{m,n}$ is the set of all lists $S = (s_1, \ldots, s_m$) such that $s_i \in \{0,1\}$ and $\sum_{i\in [m]}s_i= n$. \ 
The set $\{\ket{S}_{S \in \Lambda_{m,n}}\}$ is a basis for $n$-fermion and $m$-mode systems, which we refer to as the standard basis. \
The $m \times n$ column-orthonormal matrix $A$ describes the state 
\[
\ket{\Psi} = \sum_{S \in \Lambda_{m,n}} \det(A_S) \ket{S},
\]
where, for $S = (s_1, \ldots, s_m) \in \Lambda_{m,n}$, $A_S$ is the $n \times n$ submatrix obtained by removing row $i$ of $A$ if $s_i = 0$. \ 
Therefore, upon measurement, we observe the configuration $S \in \Lambda_{m,n}$ with probability 
\[
\abs{\braket{S | \Psi}}^2 = \abs{\det(A_S)}^2 = \det(K_S),
\]
where, for $S = (s_1, \ldots, s_m)$, $K_S$ is the $n \times n$ submatrix obtained by removing row and column $i$ of $K$ if $s_i = 0$. \
In other words, upon measurement, we observe a configuration $S \in \Lambda_{m,n}$ with probability equal to the corresponding principal minor of the kernel matrix $K$. \
The probability that any subset of $k$ modes is occupied corresponds to a principal minor of order $k$, obtained as above (remove the rows and columns of $K$ corresponding to unoccupied modes and compute the determinant of the resulting submatrix). \ 
For example, the diagonal entries $k_{ii}$ correspond to the one-mode correlations (i.e., $k_{ii}$ is the probability that mode $i$ is occupied). \ 
Passing $\ket{\Psi}$ through a unitary transformation $V \in \C^{m \times m}$ maps $K$ to $V K V^\dagger$. \
Given copies of the unknown state $\ket{\Psi}$, our goal is to output a column-orthonormal matrix $\hat{A}$ such that $\ket{\hat{\Psi}} = \sum_S \det(\hat{A}_S) \ket{S}$ is $\eps$-close to $\ket{\Psi}$ in trace distance. 

At a high level, our algorithm constructs $\hat{K}$ (an approximation of the kernel matrix), computes a decomposition $\hat{K} = \hat{A} \hat{A}^\dagger$, and then outputs $\hat{A}$. \
Our algorithm begins by measuring $O(\log(1/\delta)/\gamma^2)$ copies of the input state in the standard basis to empirically estimate the one-mode correlations of the state to accuracy $\pm \gamma$. \ 
The estimates are obtained simply by computing the average number of times each mode was occupied and are the diagonal entries of $\hat{K}$. \
One can then estimate the $(i,j)$ entry of $K$ as follows. \
Apply the beamsplitter 
\[
\frac{1}{\sqrt{2}} 
\begin{pmatrix}
1 & 1 \\ 1 & -1
\end{pmatrix}
\]
on modes $i$ and $j$, which maps the diagonal entry $k_{ii}$ to  
$\frac{1}{2}(k_{ii} + k_{jj} +2 \mathrm{Re}(k_{ij}) )$, and measure the resulting state in the standard basis. \
Repeat this $O(\log(1/\delta)/\gamma^2)$ times. \
As we did before, average the number of times mode $i$ is occupied to obtain an estimate for 
$\frac{1}{2}(k_{ii} + k_{jj} +2 \mathrm{Re}(k_{ij}) )$ to accuracy $\pm \gamma$.  \
Finally, using the previously obtained estimates for $k_{ii}$ and $k_{jj}$, solve for $\mathrm{Re}(k_{ij})$ (up to accuracy $\pm \gamma$). \
Repeat this process with the beamsplitter 
\[
\frac{1}{\sqrt{2}} 
\begin{pmatrix}
1 & i \\ 1 & -i
\end{pmatrix}
\]
to estimate the imaginary part of $k_{ij}$ to accuracy $\pm \gamma$. 

Our algorithm proceeds as follows. \ 
Simultaneously execute the process above on the pairs of modes $(1,2), (3, 4),$ $\ldots, (m-1,m)$, then $(1, 3), (2,4), \ldots, (m-2,m)$, and so on, until all the off-diagonal entries are recovered. \ It is easy to check that $O(m)$ measurement bases are needed to recover all off-diagonal entries of $K$. \
Then, we compute $Q\Lambda Q^\dagger$, an eigendecomposition of $\hat{K}$. \ 
Finally, we set $\hat{A}$ to be the $m \times n$ matrix corresponding to the first $n$ columns of $Q$, and return $\hat{A}$. \ 
Overall, the algorithm requires $O(m/\gamma^2)$ copies of the input state and $O(m^2/\gamma^2)$ time.  

The technical part is to understand how far $\ket{\hat{\Psi}}$ is from $\ket{\Psi}$ in trace distance, given that our algorithm begins by learning the entries of $K$ to within $\gamma$ in magnitude. \ 
To do this, we give a new proof that learning the kernel matrix is enough to learn the state, despite the kernel matrix only consisting of one- and two-mode correlations. \ 
\begin{theorem}[Informal version of \cref{thm:kernel-matrix-suffices}]
Let $\ket{\Psi}$ and $\ket{\hat{\Psi}}$ be $n$-fermion and $m$-mode non-interacting fermion states with kernel matrices $K$ and $\hat{K}$, respectively.
Then 
\[
d_{\rm{tr}}\left(\ket{\hat{\Psi}}, \ket{\Psi} \right) \leq \sqrt{n \norm{\hat{K} - K}_2},
\]
where $d_{\rm{tr}}(\cdot, \cdot)$ is the trace distance and $\norm{\cdot}_2$ is the spectral norm.
\end{theorem}

While this may seem surprising, this has been known to physicists since the 1960s and is (in some sense) the content of the Hohenberg-Kohn theorems \cite{hohenberg1964inhomogeneous} and Kohn-Sham equations \cite{kohn1965self}, which form the theoretical foundations of density functional theory. \ These results paved the way for computational methods in quantum chemistry, earning Walter Kohn and John Pople the 1998 Nobel Prize in Chemistry. \
Additionally, it is well-known that Wick's theorem \cite{wick1950evaluation} can be used to write higher-order correlations in terms of one- and two-mode correlations. \

Although this topic has received intense study for decades, we claim that our error analysis offers two improvements. \   
First, this is the first ``physics-free'' proof that kernel matrices suffice to learn the state: we make no mention of creation/annihilation operators, energy potentials, Hamiltonians, or the like. \ We believe our proof can be understood by any mathematician or theoretical computer scientist without a physics background and perhaps even undergraduates with a linear algebra background. \ 
Second, our theorem quantitatively relates the distance between the states and the distance between the kernel matrices, which (to our knowledge) has never been done. \

In \cref{sec:error-analysis} we use this theorem to show that the trace distance between $\hat{\ket{\Psi}}$ and $\ket{\Psi}$ is at most $\sqrt{2 n m \gamma}$. \
Therefore, the trace distance will be $\eps$-close if we set $\gamma$ to $\eps^2/2n m$.

\subsection{Related Work}
Previous work showed how to \textit{simulate} non-interacting fermions efficiently. \
In 2002, Valiant \cite{valiant2002quantum} introduced a class of quantum circuits called matchgate circuits and showed that they can be simulated classically in polynomial time. \
Soon after, Terhal and DiVincenzo \cite{terhal2002classical} (see also Knill \cite{knill2001fermionic}) showed that evolutions of non-interacting fermions give rise to unitary matchgate circuits. \
(See also \cite{aaronson2013bosonsampling}[Appendix 13] for a simpler and faster simulation algorithm.) \
Thus, the contribution of this paper is to complement these classical simulation results with an efficient \textit{learnability} result.

More broadly, \textit{quantum state tomography} is the task of constructing a classical description of a $d$-dimensional quantum mixed state, given copies of the state. \
With entangled measurements, the optimal number of copies for quantum state tomography is known to be $\Theta(d^2)$ due to Haah et al.~\cite{haah2017sample} and O'Donnell and Wright \cite{o2016efficient}. \
With unentangled measurements, the optimal number of copies is $\Theta(d^3)$ \cite{kueng2017low, haah2017sample, guctua2020fast}. \

Quantum state tomography can be computationally efficient in restricted settings. \ 
Montanaro \cite{montanaro2017learning} showed that \textit{stabilizer states} are efficiently learnable using measurements in the Bell basis. \ 
Cramer et al.\ \cite{cramer2010efficient} showed that states approximated by \textit{matrix product states} are efficiently learnable. \ Arunachalam et al.\ \cite{arunachalam2022phase} showed that some classes of \emph{phase states} are efficiently learnable. \
With this work, non-interacting fermion states is an additional class of quantum states for which we know computationally efficient learning algorithms.

Different models for learning properties of mixed states $\rho$ have been studied. \ For example, Aaronson \cite{aaronson2007learnability} showed that such states are learnable under the Probably Approximately Correct (PAC) model, using training sequences of length only logarithmic in the Hilbert space dimension. \
Since our goal is simply to reconstruct a distribution, we have no need for the PAC framework. \
Aaronson also introduced \textit{shadow tomography} \cite{doi:10.1137/18M120275X, 10.1145/3406325.3451109}, where, given a list of known two-outcome observables and copies of an unknown state, the goal is to estimate the expectation value of each observable with respect to the unknown state to additive accuracy. \ 
Although computationally inefficient, Aaronson showed that the number of copies of the input state scales logarithmically with both the number of observables and the Hilbert space dimension. \
Soon after, Huang, Kueng, and Preskill \cite{huang2020predicting} introduced \textit{classical shadows}, a shadow tomography algorithm that is computationally efficient for certain problem instances. \ For example, with random Clifford measurements, the classical time cost in classical shadows is dominated by computing quantities of the form $\braket{s|O|s}$, where $O$ is an observable and $\ket{s}$ is some stabilizer state, which is computationally efficient for certain observables. 

Recently, there have been several results that extend the classical shadows protocol to fermionic states and circuits \cite{zhao2021fermionic, wan2022matchgate, low2022classical, ogorman2022fermionic}. \
In particular, after the original version of this paper appeared but before the current version, O'Gorman \cite{ogorman2022fermionic} gave an algorithm for learning non-interacting fermion states, which is based on classical shadows. \ 
\ His learning algorithm uses $O(m^7n^2 \log(m/\delta)/\eps^4)$ samples and $O( m^9 n^2 \log(m/\delta)/\eps^4)$ time to learn a non-interacting fermion state to trace distance at most $\eps$ with probability at least $1-\delta$. \ Our work has substantially better sample and time complexities, and makes no use of randomized measurements.

Finally, \cite[Appendix C]{google2020hartree} proposes an algorithm for reconstructing a kernel matrix that involves iterating over $O(m)$ perfect matchings, just as our algorithm does, which we were unaware of until the final stages of our work. \ 
However, we note that their circuits/measurements differ from ours, and they do not provide an error analysis for their algorithm. 

\subsubsection{Determinantal Point Processes}\label{subsec:dpp}
For reasons having nothing to do with non-interacting fermions, problems extremely close to ours have already been studied in classical machine learning, in the field of \textit{Determinantal Point Processes} (DPPs). \  
A DPP is a model specified by an $m \times m$ matrix $K$ (typically symmetric or Hermitian), such that the probabilities of various events are given by various principal minors of $K$, exactly as for non-interacting fermions. \ The connection between DPPs and fermions has been known for decades \cite{macchi1975coincidence}. 

Two results in particular are directly relevant to us: Rising et al.\ \cite{rising2015efficient} and Urschel et al.\ \cite{urschel2017learning}. \
Rising et al.\ give an efficient algorithm for the \textit{symmetric principal minor assignment problem}: given a list of all $2^m$ principal minors of an unknown $m \times m$ symmetric matrix $K$, reconstruct $K$. \
Their algorithm is based on constructing an $m$-vertex graph with an edge from $i$ to $j$ whenever $K_{ij}\ne 0$, and then analyzing the minimum spanning trees and chordless cycles in that graph. \
Rising et al., however, do not do an error analysis (they assume exact knowledge of the principal minors), and they solve the problem for real and complex \emph{symmetric} matrices, whereas in our problem, the matrix is complex and Hermitian. \
This difference turns out to be surprisingly important, as the determinants of Hermitian matrices are always real, so much of the phase information vanishes---making the Hermitian case much harder. \ 

Urschel et al.\ \cite{urschel2017learning} further exploit connections between DPPs and graph theory to give an algorithm that recovers the entries of $K$, given samples from an unknown DPP. \ Their focus is on parameter learning (i.e., approximately recovering the entries of $K$), rather than learning the distribution induced by $K$ in, say, total variation distance. \ 
They again assume that the DPP is described by a real symmetric matrix. \ 

Our work provides insight on the Hermitian versions of these problems. \ 
Since there are many non-interacting fermion states with the same distribution over the standard basis, there must be many kernel matrices that are consistent with the same list of principal minors. \ The goal for the Hermitian principal minor assignment problem is to output any such matrix. \ 

It is also clear that the Hermitian version of Urschel et al.'s problem is \textit{impossible}. \ 
Samples from an unknown DPP correspond to standard basis measurement outcomes, and learning the entries of a DPP corresponds to learning the kernel matrix. \ However, in \cref{thm:kernel-matrix-suffices}, we show that learning the kernel matrix suffices to learn the entire state, which is impossible when restricted to standard basis measurements. \ 
(Even distinguishing $\ket{+}$ from $\ket{-}$ is information-theoretically impossible when given only standard basis measurements.) \ 
What one could hope for, and which we leave open, is to learn \emph{some} kernel matrix that gives rise to a distribution close in variation distance to the observed one. 

\subsubsection{Errors in Previous Version}
A previous version of this manuscript \cite{aaronson2021efficient}---where we claimed to recover a non-interacting fermion distribution using only standard basis measurements---had serious errors, which we explain below.\footnote{These errors were previously explained in \href{https://scottaaronson.blog/?p=5706}{https://scottaaronson.blog/?p=5706}.} 

In the previous manuscript, we sought to recover the rows $v_1, \ldots, v_m \in \C^n$ of the $m \times n$ column-orthonormal matrix $A$ up to isometry (see \cref{subsec:main-ideas} for the definition of $A$). \
By estimating the two-mode correlations (i.e., the probability of finding a fermion in both mode $i$ and mode $j$), 
one can deduce the approximate value of $\abs{\langle v_i, v_j\rangle}$, i.e., the absolute value of the inner product, for any $i \neq j$. \
From that information, our goal was to recover $v_1, \ldots, v_m$ (or, more precisely, their relative configuration in $n$-dimensional space up to isometry).

The approach was as follows:
if we knew $\langle v_i, v_j \rangle$ for all $i \neq j$, then we would get linear equations that iteratively constrained each $v_i$ in terms of $\langle v_i, v_j \rangle$ for $j < i$, so all that would be required would be to solve those linear systems, and then show that the solution is robust with respect to small errors in our estimates of each $\langle v_i, v_j \rangle$. \
While it is true that the measurements only reveal $\abs{\langle v_i, v_j \rangle}$ rather than $\langle v_i, v_j \rangle$ itself, the ``phase information'' in $\langle v_i, v_j \rangle$ seemed manifestly irrelevant, since it in any case depended on the irrelevant global phases of $v_i$ and $v_j$ themselves.

Alas, it turns out that the phase information \emph{does} matter. \
As an example, suppose we only knew the following about three unit vectors $u, v, w \in \R^3$:
\[
\abs{\langle u, v \rangle} = \abs{\langle u, w \rangle} = \abs{\langle w, v \rangle} = \frac{1}{2}.
\]
This is not enough to determine these vectors up to isometry! \ 
In one class of solution, all three vectors belong to the same plane, like so:
\[
u = \left(1,0,0\right),\quad v = \left(\frac{1}{2},\frac{\sqrt{3}}{2},0\right), \quad w=\left(-\frac{1}{2},\frac{\sqrt{3}}{2},0\right).
\]
In a completely different class of solution, the three vectors do not belong to the same plane, and instead look like three edges of a tetrahedron meeting at a vertex:
\[
u=\left(1,0,0 \right),\quad v=\left( \frac{1}{2}, \frac{\sqrt{3}}{2},0 \right), \quad w= \left(\frac{1}{2}, \frac{\sqrt{3}}{6}, \sqrt{\frac{2}{3}}\right).
\]
Both classes of solutions are shown in \cref{fig:my_label}. \
These solutions correspond to different sign choices for $\abs{\langle u, v \rangle}, \abs{\langle u, w \rangle},$ and $\abs{\langle w, v \rangle}$---choices that collectively matter, even though each one is individually irrelevant.

\begin{SCfigure}[1.95]
    \centering
    \includegraphics[width=.28\textwidth]{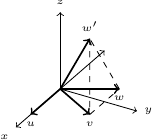}
    \caption{The vectors $u, v,$ and $w$ belong to the same plane, while $v,w,$ and $w^\prime$ are edges of a tetrahedron that meet at a vertex. \
    Both sets of vectors satisfy the same inner product constraints \emph{up to phase information}, yet belong to two distinct classes of solutions.}
    \label{fig:my_label}
\end{SCfigure}

It follows that, even in the special case where the vectors are all real, the two-mode correlations are not enough to determine the vectors' relative positions. \
And alas, the situation is even worse when, as for us, the vectors can be complex. \
Indeed, even for systems of $2$ fermions and $4$ modes, it is possible to exhibit distributions that require complex vectors. \
For example, let
\[
A = \begin{pmatrix}
\frac{1}{2} & 0 \\
\frac{1}{2\sqrt 2}  & \sqrt{\frac{2}{5}} \\
\frac{1}{2\sqrt 2} & \sqrt{\frac{2}{5}}i \\
\frac{1}{\sqrt 2} & -\frac{1}{\sqrt{10}} - \frac{1}{\sqrt{10}}i
\end{pmatrix}.
\]
Denote the one-mode correlations by $p_i$ and the two-mode correlations by $p_{ij}$. \
Then by explicit calculation (e.g., compute $K$, then compute the appropriate principal minors of $K$), one can verify that
\[
p_1 = \frac{1}{4}, \quad p_2 = p_3 = \frac{21}{40}, \quad p_{12} = p_{13} = p_{23} = \frac{1}{10}.\footnote{We do not need to look at the fourth mode to show that complex numbers are necessary.}
\]
Hence, to represent the corresponding distribution with real vectors only, one must find three vectors in $\R^2$ with squared lengths $\frac{1}{4}, \frac{21}{40},$ and $\frac{21}{40}$, respectively, such that the squared area of the parallelogram created by any pair of vectors is $\frac{1}{10}$. \ One can verify that this is not possible. 

We conclude that any possible algorithm for learning fermionic distributions from standard-basis
measurements will have to solve a system of nonlinear equations (albeit, a massively overconstrained
system that is guaranteed to have a solution); it will have to use three-mode correlations (i.e.,
statistics of triples of fermions), and indeed (one can show) in some exceptional cases four-mode
correlations and above; it will sometimes have to output complex solutions even when all the
input data is real (which rules out a purely linear-algebraic approach); and 
it will have to learn the phase information relevant to the \emph{distribution} (rather than the entire state). 

Finally, separate from the issues above, we are grateful to Andrew Zhao for identifying an error in the most recent version of this manuscript (v3 on the arXiv). \ In our error analysis, we assumed that the $m \times m$ matrix output by our algorithm is rank-$n$, when, in fact, it is full-rank. \ In short, we handle this issue by computing an eigendecomposition of said matrix and discarding the smallest $m - n$ eigenvalues and eigenvectors. \ 

\section{Preliminaries}
Throughout this work, we use the following notation. \
$[n] \coloneqq \{1,\ldots,n\}$. \  
Let $X \in \C^{n \times n}$ and $v \in \C^n$. \
Then $\norm{v}_p \coloneqq (\sum_{i \in [n]}\abs{v_i}^p)^{1/p}$ is the $\ell_p$-norm, and $\norm{X}_{2} \coloneqq \sup_{\norm{v}_2 = 1} \norm{Xv}_2$ is the spectral norm. \
Let $\rho$ and $\sigma$ be two quantum mixed states. \
Then $d_\mathrm{{tr}}(\rho, \sigma) \coloneqq \frac{1}{2} \mathrm{tr}\left( \sqrt{(\rho-\sigma)^2}\right) = \frac{1}{2} \sum_i \abs{ \lambda_i}$ is the trace distance, where the $\lambda_i$'s are eigenvalues of the error matrix $\rho-\sigma$.

\subsection{Non-Interacting Fermions}

We briefly review non-interacting fermions. \ For more detail, see e.g.\ \cite[Chapters 7-10]{aaronson2022introduction}. 

As mentioned above, $\Lambda_{m,n}$ is the set of all lists $S = (s_1, \ldots, s_m$) such that $s_i \in \{0,1\}$ and $\sum_{i\in [m]}s_i= n$. \
One can easily verify that the total number of configurations $\abs{\Lambda_{m,n}} = \binom{m}{n}$. Define $M \coloneqq \binom{m}{n}$. \
The set $\{\ket{S}\}_{S \in \Lambda_{m,n}}$ is a basis for systems of $n$ non-interacting fermions in $m$ modes, which we refer to as the standard basis (also called the Fock basis or occupation number basis). \
Hence, a non-interacting fermion state is a unit vector in an $M$-dimensional complex Hilbert space with the form
\[
\ket{\Psi} = \!\! \sum_{S \in \Lambda_{m,n}}\!\!\!\! \alpha_S \ket{S}, \quad \text{where $\sum_{S \in \Lambda_{m,n}} \abs{\alpha_S}^2 = 1$},
\]
and, upon measurement, we observe an element $S \in \Lambda_{m,n}$ with probability 
\[
\Pr[S] = \abs{\braket{\Psi | S}}^2 = \abs{\alpha_S}^2.
\]

The transformations on non-interacting fermion states can be described by $m \times m$ unitary matrices, which can always be constructed with $O(m^2)$ elementary operations called \textit{beamsplitters} and \textit{phaseshifters} (see \cite{PhysRevLett.73.58} for a proof of this statement). \
A beamsplitter acts on two modes and has the form 
\[
\begin{pmatrix}
  1 & & & & & \\
  & \ddots & & & & \\
&& \cos \theta &-\sin \theta && \\
&&\sin \theta & \cos \theta && \\
  & & & & \ddots & \\
  & & & &  & 1 \\
\end{pmatrix},
\]
while a phaseshifter applies a complex phase to a single mode and has the form 
\[
\begin{pmatrix}
  1 & & & & \\
  & \ddots & & & \\
  &  & e^{i \theta} & & \\
  & & & \ddots & \\
  & & &  & 1 \\
\end{pmatrix}.
\]

An $m \times m$ unitary $U$ describes how a single fermion would be scattered from one mode to the others. \
This induces an $M \times M$ unitary transformation $\varphi(U)$ on the Hilbert space of $n$-fermion states, where $\varphi$ is the homomorphism that lifts $m \times m$ unitary transformations to $n$-fermion unitary transformations. \
For an $m \times m$ unitary matrix $U$, one way to define $\varphi$ is
\begin{align}\label{eq:hom}
\bra{S}\varphi(U)\ket{T} = \det(U_{S,T}),
\end{align}
for all $S = (s_1, \ldots, s_m), T = (t_1, \ldots, t_m) \in \Lambda_{m,n}$, where $U_{S, T}$ is the $n\times n$ submatrix obtained by removing row $i$ of $U$ if $s_i = 0$ and removing column $j$ of $U$ if $t_j = 0$. \
One can also view $\varphi(U)$ as a unitary transformation on $m$-qubit states. \ In this case, $\varphi(U)$ is a Hamming-weight-preserving matchgate unitary \cite{valiant2002quantum}. \
This follows directly from Terhal and DiVincenzo's work \cite{terhal2002classical}, which shows the equivalence between evolutions of non-interacting fermions and matchgate circuits. 

Intuitively, the reason the determinant arises in \cref{eq:hom} is that $\bra{S} \varphi(U) \ket{T}$ is the sum over the $n!$ permutations that take $n$ fermions in configuration $S$ to configuration $T$, each permutation contributing to the overall amplitude. \
When a permutation from $S$ to $T$ is odd, its contribution to the overall amplitude has a phase factor of $-1$, while when the permutation is even the phase factor is $1$. \
This is the antisymmetry property of fermions: swapping two fermions picks up a $-1$ phase in the amplitude.\footnote{
Meanwhile, bosons are symmetric under transpositions, so no minus signs show up. This is precisely why permanents arise when computing amplitudes for non-interacting bosons, while determinants show up for non-interacting fermions.}

\subsection{Problem Setup}
We use the following notation: 
$\ket{1_n} \coloneqq \ket{1, \ldots, 1, 0, \ldots, 0}$ is the standard initial state (where $n$ fermions occupy the first $n$ modes), $U \in \C^{m \times m}$ is the unitary that prepares the unknown non-interacting fermion state (i.e., $\ket{\Psi} = \varphi(U)\ket{1_n}$), and $A \in \C^{m \times n}$ is  the $m \times n$ column-orthonormal matrix corresponding to the first $n$ columns of $U$. \
Define $K \coloneqq A A^\dagger$. 

For each element $S = (s_1, \ldots, s_m) \in \Lambda_{m,n}$, let $A_S$ be the $n \times n$ submatrix obtained by removing row $i$ of $A$ if $s_i = 0$, and let $K_S$ be the $n \times n$ submatrix obtained by removing row and column $i$ of $K$ if $s_i = 0$. \
Then from \cref{eq:hom}, it follows that 
\[
\ket{\Psi} = \sum_{S \in \Lambda_{m,n}} \det(A_S) \ket{S}, 
\]
and $\calD_K$, the probability distribution over $S \in \Lambda_{m,n}$ obtained by measuring the state $\ket{\Psi} = \varphi(U)\ket{1_n}$ in the standard basis, is given by 
\[
\abs{\bra{1_n}\varphi(U)\ket{S}}^2 = \abs{\det(A_S)}^2 = \det(K_S),
\]
where the last equality
uses the fact that, for any square matrices $X$ and $Y$, $\det(X)^* = \det(X^\dagger)$ and $\det(X)\det(Y) = \det(XY)$. \
Further, for any list $S = (s_1, \ldots, s_m)$ where $\sum_i s_i = k < n$, the marginal probability that those $k$ modes are occupied is $\det(K_S)$.

It is easy to verify that $K$ is Hermitian, positive semi-definite, and a projector ($K^2 = K$); and that $\tr(K) = n$. \
Additionally, observe that the $(i,j)$ entry is the inner product between the $i$th and $j$th rows of $A$ (i.e., $K$ is a \textit{Gram matrix}). \
Finally, note that $A$ is a \emph{highly non-unique} description of $\ket{\Psi}$ and $K$. \ Let $R$ be any $n \times n$ unitary matrix. \ Then $A$ and $A R$ describe the same state:
\[
\ket{\Psi} = \sum_{S \in \Lambda_{m,n}} \det(A_SR) \ket{S}
= \sum_{S \in \Lambda_{m,n}} \det(A_S) \det(R) \ket{S}
= \sum_{S \in \Lambda_{m,n}} \det(A_S)\ket{S},
\]
where, in the last equality, we use the fact that the determinant of a unitary matrix is a complex unit, which only adds an irrelevant global phase. \
Meanwhile, the kernel matrix of $\ket{\Psi}$ is unchanged: $K = (AR)(AR)^\dagger = AA^\dagger$.

Applying a unitary $V$ maps $\ket{\Psi} = \varphi(U) \ket{1_n}$ to $\varphi(VU) \ket{1_n}$. \
It is easy to check that the probability that we observe $S$ upon measuring $\varphi(VU)\ket{1_n}$ in the standard basis is 
\[
 \abs{\bra{1_n}\varphi(VU)\ket{S}}^2 = \det( (VKV^\dagger)_S),
\]
and in general, applying $V$ has the following effect on the matrix $K$: $K \mapsto V K V^\dagger$. 

Given copies of $\ket{\Psi}$ and the ability to apply beamsplitter networks before measurement, 
our goal is to output a matrix $\hat{A} \in \C^{m \times n}$ such that $\ket{\hat{\Psi}} = \sum_S \det(\hat{A}_S)$ is $\eps$-close to $\ket{\Psi}$ in trace distance.

\section{Learning Algorithm}\label{sec:learning-algorithm}
In this section, we present our learning algorithm, which is given copies of an unknown non-interacting fermion state, and outputs an $m \times n$ matrix $\hat{A}$ such that the corresponding state $\ket{\hat{\Psi}} = \sum_{S} \det(\hat{A}_S) \ket{S}$ is close to the original in trace distance. \
The algorithm has three phases: first, we learn the diagonal entries of the kernel matrix $K$ with standard basis measurements, then we learn the off-diagonal entries by measuring the unknown state in $O(m)$ different bases. \ Finally, we decompose the reconstructed kernel matrix into the $m \times n$ output matrix $\hat{A}$.

\begin{myalgorithm}[label={learning-algorithm}]{Efficient tomography of non-interacting fermion states}
\textbf{Input:}
Black-box access to copies of $\ket{\Psi}$ (the input state), $\gamma \in (0,1)$ (accuracy parameter), and $\delta \in (0,1)$ (confidence parameter).

\textbf{Output:} an $m \times n$ matrix $\hat{A} \in \C^{m \times n}.$\vspace{.09in}

\begin{enumerate}[leftmargin=*,label={\arabic*:}]
\item Measure $O(\log(1/\delta)/\gamma^2)$ copies of $\ket{\Psi}$ in the standard basis and estimate the one-mode correlations for all $m$ modes. \ Set $\hat{k}_{ii}$ to the empirical estimate that mode $i$ is occupied.
\item Choose $O(m)$ different perfect matchings among the $m$ modes, which together cover all possible $(i,j)$ pairs.  
\item \textbf{for each} perfect matching \textbf{do}
\item $\quad$ Apply the beamsplitter 
\[
\frac{1}{\sqrt{2}} 
\begin{pmatrix}
1 & 1 \\ 1 & -1
\end{pmatrix}
\]
to each pair of modes in the perfect matching and measure in the standard basis. \ Repeat this $O(\log(1/\delta)/\gamma^2)$ times, and use the measurement data to estimate the one-mode correlations.  
\item $\quad$ Repeat the previous step with the beamsplitter
\[
\frac{1}{\sqrt{2}} 
\begin{pmatrix}
1 & i \\ 1 & -i
\end{pmatrix}.
\]
\item $\quad$
For each pair $(i,j)$ in the perfect matching, 
the beamsplitter network in step 4 maps the one-mode correlation $k_{ii}$ to 
$k^\prime_{ii} \coloneqq \frac{1}{2}(k_{ii} + k_{jj} +2 \mathrm{Re}(k_{ij}))$. \
Denote the $i$th estimate obtained in step 4 by $\hat{k}^\prime_{ii}$. \ 
Let $\mathrm{Re}(\hat{k}_{ij})$ be the estimate of $\mathrm{Re}(k_{ij})$ obtained by solving the following equation:
\[
\hat{k}^\prime_{ii} = \frac{1}{2}(\hat{k}_{ii} + \hat{k}_{jj} +2 \mathrm{Re}(\hat{k}_{ij})),
\]
where $\hat{k}_{ii}$ and $\hat{k}_{jj}$ are the estimates from step 1. 
\item $\quad$
Repeat the previous step with the estimates from step 5 to obtain estimate $\mathrm{Im}(\hat{k}_{ij})$ for each pair $(i,j)$. \
(Note that the beamsplitter network in step 5 maps $k_{ii}$ to $\frac{1}{2}(k_{ii} + k_{jj} + 2 \mathrm{Im}(k_{ij}) )$, for each pair $(i,j)$ in the perfect matching.)
\item $\quad$
For each pair $(i,j)$ in the perfect matching, set $\hat{k}_{ij} = \mathrm{Re}(\hat{k}_{ij}) +   i \mathrm{Im}(\hat{k}_{ij})$ and $\hat{k}_{ji} = \hat{k}_{ij}^*$. 
\item Let $\hat{K} = (\hat{k}_{ij}) \in \C^{m \times m}$, and let $Q \Lambda Q^\dagger$ be an eigendecomposition of $\hat{K}$. \ Set $\hat{A}$ to be the $m \times n$ matrix corresponding to the first $n$ columns of $Q$.
\item \Return $\hat{A}$. 
\end{enumerate}
\end{myalgorithm}

For each measurement basis, we estimate $O(m)$ entries of $K$ to within $\gamma$ in magnitude, which can be accomplished with $O(\log(1/\delta)/\gamma^2)$ copies \cite[Theorem 9]{canonne2020short} and $O(m \log(1/\delta)/\gamma^2)$ time. \ 
To estimate the off-diagonal entries to the target accuracy, an additional constant factor appears in the sample complexity, which is absorbed into the $O(\log(1/\delta)/\gamma^2)$. \
We use $O(m)$ measurement bases in total, so the overall copy and time complexities are $O(m \log(1/\delta) / \gamma^2)$ and $O(m^2 \log(1/\delta) / \gamma^2)$ respectively. \
Our algorithm then computes an eigendecomposition of an $m \times m$ matrix, which requires $O(m^3)$ time, but computing this decomposition is not the bottleneck in our algorithm.

We note that $\hat{K}$ is clearly Hermitian by construction, so the eigendecomposition of $\hat{K}$ exists. \ Following convention, it is assumed that the eigenvalues are ordered from largest to smallest (i.e., the first column of $Q$ corresponds to the largest eigenvalue of $\hat{K}$, and so on). \ 
The output of our algorithm $\hat{A}$ is column-orthonormal because the columns of the unitary $Q$ are orthonormal. \ 
Therefore, $\hat{A}$ describes a non-interacting fermion state $\ket{\hat{\Psi}} = \sum_S \det(\hat{A}_S) \ket{S}$. \
In the next section, we show that if $\gamma = \frac{\eps^2}{2nm }$, then the trace distance between $\ket{\Psi}$ and $\ket{\hat{\Psi}}$ is at most $\eps$. \
Hence, for the two states to be $\eps$-close in trace distance, $ O(m^3 n^2 \log(1/\delta) / \eps^4)$ copies and $O(m^4 n^2 \log(1/\delta)/\eps^4)$ time suffice.
 
\section{Error Analysis}\label{sec:error-analysis}
In this section, we show that the trace distance between $\ket{\Psi}$ and $\ket{\hat{\Psi}}$ is at most $\sqrt{2 n m \gamma}$. \ 
Therefore, if $\gamma = \frac{\eps^2}{2nm}$, then the trace distance between $\ket{\Psi}$ and $\ket{\hat{\Psi}}$ is at most $\eps$. \

The error analysis is presented in two parts. \  
First, we show that the trace distance between any two non-interacting fermion states is bounded above the spectral distance between their kernel matrices. \ Then we show that the output of our algorithm is close to the original state in trace distance. 

\subsection{The Kernel Matrix Suffices}
We prove that the trace distance between two non-interacting fermion states is upper bounded by the spectral difference between their kernel matrices. \ To show this, we need the following two lemmas.

\begin{lemma}\label{fact:inductive}
Let $a_1, a_2, \ldots, a_n \in [0,1]$. \ Then 
\[ 
1 - \prod_i a_i \leq n \max_i 1 - a_i.
\]
\end{lemma}
\begin{proof}
For any $x, y \in [0,1]$,  
\[
1 - xy = 1 - x + x - xy = (1 - x) + x(1 - y) \leq (1-x) + (1-y).
\]
We can inductively apply this to get 
\[
1 - \prod_i a_i \leq \sum_i 1 - a_i \leq n \max_i 1 - a_i.\qedhere
\]
\end{proof}

\begin{lemma}\label{fact:kernel}
Let $A, \hat{A} \in \C^{m \times n}$ $(m \geq n)$ be $m \times n$ matrices, and let $A$ be column-orthonormal. \ 
Define $K \coloneqq A A^\dagger$ and $\hat{K} \coloneqq \hat{A}\hat{A}^\dagger$. \
Let $\Sigma  = \mathrm{diag}(\sigma_1, \ldots, \sigma_n)$ where $\sigma_i$ are the singular values of $\hat{A}^\dagger A$. 
Then
\[
\norm{I - \Sigma^2}_2 \leq \norm{\hat{K} - K}_2.
\]
\end{lemma}
\begin{proof}
Since $A$ is column-orthonormal, $A^{\dagger}A = I$ ($I$ is the identity matrix). \
Let $Q \Sigma V^\dagger$ be a singular value decomposition of $\hat{A}^\dagger A$. \ First, note that 
\[
V \Sigma^2 V^\dagger 
= (Q \Sigma V^\dagger)^\dagger Q \Sigma V^\dagger 
= (\hat{A}^\dagger A)^\dagger \hat{A}^\dagger A 
= A^\dagger \hat{A} \hat{A}^\dagger A.
\]
Therefore, 
\begin{align*}
 \norm{I - \Sigma^2}_{2}
=  \norm{I - A^\dagger \hat{A} \hat{A}^\dagger A}_2
=  \norm{A^\dagger A - A^\dagger \hat{A} \hat{A}^\dagger A}_2
=  \norm{A^\dagger (A - \hat{A} \hat{A}^\dagger A)}_2
\leq \norm{A - \hat{A} \hat{A}^\dagger A}_2,
\end{align*}
where the first step uses the fact that the spectral norm is unitarily invariant and the final step follows from the submultiplicativity of matrix norms and that $\norm{A^\dagger}_2 = \norm{A}_2 = 1$ since $A$ is column-orthonormal. 
Finally, 
\[
\norm{A - \hat{A} \hat{A}^\dagger A}_2 
= \norm{AA^\dagger A - \hat{A} \hat{A}^\dagger A}_2 
\leq \norm{AA^\dagger  - \hat{A} \hat{A}^\dagger }_2 
= \norm{K - \hat{K}}_2.\qedhere
\]
\end{proof}

We are now ready to show that if two kernel matrices are close, then the corresponding states will also be close.

\begin{theorem}\label{thm:kernel-matrix-suffices}
Let $\ket{\Psi}$ and $\ket{\hat{\Psi}}$ be non-interacting fermion states of $n$ fermions and $m$ modes described by the $m \times n$ column-orthonormal matrices $A, \hat{A} \in \C^{m \times n}$, respectively. \
Define $K \coloneqq A A^\dagger$ and $\hat{K} \coloneqq \hat{A}\hat{A}^\dagger$. \
Then
\[
d_{\rm{tr}}\left(\ket{\hat{\Psi}}, \ket{\Psi} \right) \leq \sqrt{n \norm{\hat{K} - K}_2}.
\]
\end{theorem}
\begin{proof}
Recall that $\ket{\Psi}$ and $\ket{\hat{\Psi}}$ can be written as
\[
\ket{\Psi} = \sum_{S \in \Lambda_{m,n}} \det(A_S) \ket{S} \qquad \text{and} \qquad  \ket{\hat{\Psi}} = \sum_{S \in \Lambda_{m,n}} \det(\hat{A}_S)  \ket{S}
\]
for the column-orthonormal matrices $A, \hat{A} \in \C^{m \times n}$. \ 
Then
\begin{align*}
d_{\rm{tr}}\left(\ket{\hat{\Psi}}, \ket{\Psi} \right)
&= \sqrt{1 - \abs{\braket{\hat{\Psi}|\Psi}}^2} \\
&= \sqrt{1 - \Abs{\sum_{S\in \Lambda_{m,n}} \det(\hat{A}_S)^* \det(A_S) }^2} \\
&= \sqrt{1 - \Abs{\sum_{S\in \Lambda_{m,n}} \det(\hat{A}_S^\dagger) \det(A_S) }^2} \\
&= \sqrt{1 - \abs{ \det(\hat{A}^\dagger A)}^2}, \\
\end{align*}
where the second-to-last step follows because, for any square matrix $X$, $\det(X^T) = \det(X)$ and $\det(X)^* = \det(X^*)$; and the final step follows from the Cauchy-Binet formula.

Let $Q \Sigma V^\dagger$ be a singular value decomposition of $\hat{A}^\dagger A$, where $\sigma_1, \ldots, \sigma_n$ are the singular values on the diagonal of $\Sigma$. \ Then
\begin{align*}\label{eq:det}
\abs{\det(\hat{A}^\dagger A) }^2 
= \abs{\det(Q \Sigma V^\dagger) }^2 
= \abs{\det(Q) \det(\Sigma) \det(V^\dagger) }^2 
= \det(\Sigma)^2
= \prod_i \sigma_i^2.
\end{align*}
Note that, for all $i \in [n]$, $\sigma_i \leq \sigma_i(\hat{A}^\dagger) \norm{A}_2 = 1$ \cite[Chapter 3]{horn1994topics}, where $\sigma_i(\hat{A}^\dagger)$ is the $i$th singular value of $\hat{A}^\dagger$. \ 
Plugging this into our bound on the trace distance, we get 
\begin{align*}
d_{\rm{tr}}\left(\ket{\hat{\Psi}}, \ket{\Psi} \right)
&= \sqrt{1 - \abs{ \det(\hat{A}^\dagger A)}^2} \\
&= \sqrt{1 - \prod_i \sigma_i^2 }\\
&\leq \sqrt{n \left(\max_i 1 - \sigma_i^2\right)} && \text{(By \cref{fact:inductive}).} \\
&= \sqrt{n \norm{I - \Sigma^2}_2} \\
&\leq \sqrt{n \norm{\hat{K} - K}_2} && \text{(By \cref{fact:kernel}).}\qedhere
\end{align*}
\end{proof}

\subsection{Completing the Analysis}
We prove that the fermionic state output by \cref{learning-algorithm} is close to the input state in trace distance. \ 
To do so, we make use of Weyl's inequality, which implies that the spectrum of a Hermitian matrix is stable under small perturbations. \ 

\begin{theorem}[A Consequence of Weyl's Inequality \cite{weyl1912asymptotische}]\label{thm:weyl-inequality}
Let $M, N, R \in \C^{n \times n}$ be $n \times n$ Hermitian matrices such that $M = N + R$. \
Let $\lambda_1, \ldots, \lambda_n$ be the eigenvalues of $M$, and let $\mu_1, \ldots, \mu_n$ be the eigenvalues of $N$. \ Then, for all $i \in [n]$, 
\[
\abs{\lambda_i - \mu_i} \leq \norm{R}_2.
\]
\end{theorem}

We are now ready to prove that \cref{learning-algorithm} successfully learns a non-interacting fermion state.

\begin{theorem}
Let $\hat{A}$ be the output of \cref{learning-algorithm} when given $\ket{\Psi}$ as input, and let $\ket{\hat{\Psi}}$ be the non-interacting fermion state described by $\hat{A}$. \ Then 
\[
d_{\rm{tr}}\left(\ket{\hat{\Psi}}, \ket{\Psi} \right) \leq \sqrt{2 nm \gamma}. 
\]
\end{theorem}
\begin{proof}
It is convenient to recall the last steps of \cref{learning-algorithm}:

In our algorithm, once the quantum measurements are complete, we have a matrix $\hat{K}$ whose entries are within $\gamma$ in magnitude of $K$. \
We then compute the eigendecomposition $Q \Lambda Q^\dagger$ of $\hat{K}$, where the first column of $Q$ is the eigenvector corresponding to the largest eigenvalue of $\hat{K}$ and so on. \ 
Finally, we set $\hat{A}$ to be the $m \times n$ matrix corresponding to the first $n$ columns of $Q$, and output $\hat{A}$. \ 
Therefore, the kernel matrix of $\ket{\hat{\Psi}}$ is $\hat{A}\hat{A}^\dagger$, and, 
by \cref{thm:kernel-matrix-suffices}, the trace distance between $\ket{\Psi}$ and $\ket{\hat{\Psi}}$ is bounded above by $\sqrt{n \norm{\hat{A}\hat{A}^\dagger - K}_2}$, where $K = A A^\dagger$ is the kernel matrix corresponding to the input state $\ket{\Psi}$. \
To complete the proof, we must bound $\norm{\hat{A}\hat{A}^\dagger - K}_2$. \ To that end, by the triangle inequality, 
\begin{align}\label{eq:first-bound-after-triangle}
    \norm{\hat{A}\hat{A}^\dagger - K}_2 
    &\leq \norm{\hat{A}\hat{A}^\dagger - \hat{K}}_2 + \norm{\hat{K} - K}_2.
\end{align}

Observe that $\hat{K} = K + E$, where $E$ is a perturbation of $K$ whose entries have magnitude at most $\gamma$. \
Therefore, $\norm{\hat{K} - K}_2 = \norm{K + E - K}_2 = \norm{E}_2$. \
The error matrix $E$ is Hermitian because $K$ and $\hat{K}$ are Hermitian, and Hermitian matrices are closed under addition/subtraction. \
Therefore, since $\hat{K}$, $K$, and $E$ are all Hermitian, we can use \cref{thm:weyl-inequality} to upper bound the absolute difference between the eigenvalues of $K$ and $\hat{K}$. \ In particular, the absolute difference between the $i$th eigenvalues of $K$ and $\hat{K}$ is at most $\norm{E}_2$, for all $i \in [m]$. \

Let $\mathbb{1}_{n} = \mathrm{diag}(1, \ldots, 1, 0, \ldots, 0) \in \R^{m \times m}$ be the diagonal matrix whose first $n$ diagonal entries are $1$ and the rest $0$. Observe that $\hat{A}\hat{A}^\dagger = Q \mathbb{1}_n Q^\dagger$, since $\hat{A}$ is the first $n$ columns of $Q$, and recall that $\hat{K}= Q \Lambda Q^\dagger$. \ 
Therefore, 
\[
\norm{\hat{A}\hat{A}^\dagger - \hat{K}}_2 = \norm{Q \mathbb{1}_n Q^\dagger - Q \Lambda Q^\dagger}_2 = \norm{\mathbb{1}_n - \Lambda}_2, 
\]
where, in the last equality, we use the fact that the spectral norm is unitarily invariant. \
Note that $\Lambda$ contains the eigenvalues of $\hat{K}$, and $\mathbb{1}_n$ contains the eigenvalues of $K$ (since $K$ is a trace-$n$, rank-$n$ projector). \
Therefore, $\norm{\mathbb{1}_n - \Lambda}_2$ is the maximum absolute difference between the eigenvalues of $K$ and $\hat{K}$, which is at most $\norm{E}_2$, as we argued in the previous paragraph. 

The trivial bound on the spectral norm of $E$ is the Frobenius norm of $E$, which is maximum when all entries of $E$ have magnitude $\gamma$. \ Specifically, $\norm{E}_2 \leq \norm{E}_F \leq m \gamma$. \
Note $\norm{E}_2 = \norm{E}_F$ when the entries of $E$ are all $\gamma$, so we cannot hope for a tighter bound on $\norm{E}_2$.

Plugging this into \cref{eq:first-bound-after-triangle},
\begin{align*}
    \norm{\hat{A}\hat{A}^\dagger - K}_2 
    &\leq \norm{\hat{A}\hat{A}^\dagger - \hat{K}}_2 + \norm{\hat{K} - K}_2  \\
    &\leq \norm{E}_2 + \norm{E}_2 \\ 
    &\leq 2 m \gamma.\qedhere
\end{align*}
\end{proof}

\section{Connections to Quantum State Tomography}\label{sec:tomography}
Although physically different, our problem is closely related to the quantum state tomography problem. \
In quantum state tomography, we want to recover an unknown Hermitian matrix, namely a $d$-dimensional mixed state $\rho \in \C^{d \times d}$, and applying a quantum circuit $V$ to $\rho$ maps $\rho$ to $V \rho V^\dagger$. \
In our problem, applying a unitary $V$ to $\ket{\Psi}$ maps $K$ to $V K V^\dagger$, where $K$ is an unknown Hermitian matrix, and our algorithm is able to recover the entries of $K$ to within $\gamma$ in magnitude. \
Therefore, our algorithm can also be viewed as a state tomography algorithm: measure copies of $\rho$ in the $O(d)$ measurement bases obtained by choosing perfect matchings that cover all $(i,j)$ pairs and output the resulting matrix $\hat{\rho}$ (skipping the last few steps of the algorithm that involve computing an eigendecomposition). \ 
As before, the algorithm requires $O(d \log(1/\delta)/\gamma^2)$ copies and $O(d^2 \log(1/\delta)/\gamma^2)$ time. 

The error analysis is slightly different than for learning a fermionic state. \ For the state tomography problem, we want the output matrix $\hat{\rho}$ to be close to $\rho$ in trace distance, which is proportional to $\norm{\hat{\rho} - \rho}_1$, whereas, for fermionic tomography, our trace distance upper bound is proportional to $\norm{\hat{K} - K}_2$ (see \cref{thm:kernel-matrix-suffices}).  \ Hence, to analyze the performance of our algorithm for state tomography, we must upper bound $\norm{\hat{\rho} - \rho}_1$.

Recall that the error matrix $E = \hat{\rho} - \rho$ is Hermitian and has entries with magnitude at most $\gamma$. \ For $i \in [m]$, let $\lambda_i$ denote the eigenvalues of $E$. \
Then 
\[
d_{\rm{tr}}(\hat{\rho}, \rho) = \frac{1}{2} \norm{\hat{\rho} - \rho}_1 = \frac{1}{2} \norm{E}_1 = \frac{1}{2} \sum_i \abs{\lambda_i}
\leq \frac{\sqrt{d}}{2} \sqrt{\sum_i \abs{\lambda_i}^2} \leq \frac{1}{2} d^{3/2} \gamma.
\]
The first inequality follows from the fact that the arithmetic mean is bounded above by the quadratic mean, and the second inequality follows from the fact that $\sqrt{\sum_i \abs{\lambda_i}^2} = \norm{E}_F \leq d \gamma$. \
For $d_{\rm{tr}}(\hat{\rho}, \rho) \leq \eps$, we must set $\gamma = 2 d^{-3/2} \eps$. \ 
The resulting copy complexity is $O(d^4 \log(1/\delta)/\eps^2)$ and time complexity is $O(d^5 \log(1/\delta)/\eps^2)$. \
Note that the optimal copy complexity for quantum state tomography \emph{with unentangled measurements} is $\Theta(d^3/\eps^2)$ \cite{kueng2017low, haah2017sample, guctua2020fast}, compared to $\Theta(d^2/\eps^2)$ with entangled measurements \cite{o2016efficient, haah2017sample}. \ 

Finally, we show that our upper bound on $\norm{\hat{\rho} - \rho}_1$ is tight. \ 
Let $F$ be the $d \times d$ Fourier transform whose $(i,j)$ entry is $\exp(2\pi i j/d)/\sqrt{d}$. Then  $F$ scaled by a factor of $\sqrt{d}\gamma$ is a valid error matrix whose $1$-norm is equal to $d^{3/2}\gamma$. \ Indeed, any $d \times d$ unitary matrix scaled by a factor of $\gamma \sqrt{d}$ will match our upper bound.

\section{Open Problems}\label{sec:open-problems}
Perhaps the most interesting open problem is to give an algorithm to learn fermionic distributions using only standard basis measurements. \
Specifically, the following problems remain open: 

\begin{enumerate}
    \item \textbf{Learn real DPPs in variation distance.} Given sample access to a distribution induced by an $m \times m$ symmetric matrix $K \in \R^{m \times m}$, output an $m \times m$ matrix $\hat{K}$ such that the induced distribution is close in variation distance. \ (See \cref{subsec:dpp} for detail on DPPs.)
    \item \textbf{Hermitian principal minor assignment problem.} \ 
    Given a list of all $2^m$ principal minors of an unknown Hermitian matrix $K \in \C^{m \times m}$, reconstruct any Hermitian matrix that is consistent with that list. 
    \item \textbf{Learn non-interacting fermion distributions with standard basis measurements.} \ Given sample access to a non-interacting fermion distribution, efficiently learn the distribution in total variation distance. 
\end{enumerate}

The third problem is in some sense a combination of the first and second. \ To solve the first problem, we believe that the connections between DPPs and graph theory used in Rising et al.\ \cite{rising2015efficient} and Urschel et al.\ \cite{urschel2017learning} should be enough to develop an efficient algorithm. \ 
As discussed in \cref{subsec:dpp}, our work shows that there are many kernel matrices that have the same principal minors (indeed, any set non-interacting states that have the same distribution over the standard basis will give rise to a set of kernel matrices that are consistent with the same list of principal minors). \ The goal for the second and third problems is to output any one of the valid matrices. 

For the second problem, however, the following example shows that some combinatorial information about the kernel matrix $K$, above and beyond the obvious complex conjugation ambiguities, is \emph{not} determined even in principle by $K$'s principal minors. \ 
Consider the following $4 \times 4$ Hermitian matrix:
\[
K = 
\begin{pmatrix}
1&  1 & 1 & 1 \\
1 & 1 & a & b^* \\
1 & a^* &1 & c \\
1 & b & c^* &1
\end{pmatrix}.
\]
Suppose we have learned, by looking at the $2\times 2$ and $3\times 3$ principal minors, that
\[
a=e^{ix},\quad
b=e^{iy},\quad \text{and}\quad
c=e^{iz},
\]
where $\abs{x} = \abs{y} = \abs{z} =w$ for some $w$ that is known. \
That is, we have determined $a,b,$ and $c$ up to complex conjugation, and up to complex conjugation they are all equal. \
By looking at the bottom-most $3 \times 3$ principal minor, we can learn $\mathrm{Re}(abc)$ and hence $\abs{x + y + z}$. \
Suppose that this is also $w$. \ 
From the $4\times 4$ minor, combined with the $2\times 2$ and $3 \times 3$ minors, we get one additional piece of information, namely:
\[ 
\mathrm{Re}(ab) + \mathrm{Re}(ac) + \mathrm{Re}(bc).
\]
Suppose that, as expected, this is $2+\cos(2w)$. \
Then even though we have extracted all information from the principal minors, there are \textit{still} three essentially different solutions possible. Namely, 
\[
(1)\,\,x=y=w \text{ and } z=-w, \qquad
(2)\,\,x=z=w \text{ and } y=-w, \qquad
(3)\,\,y=z=w \text{ and } x=-w.
\]

Of course, for fermionic distributions, the matrix $K$ must be Hermitian, positive semi-definite,  and a projector, and $\mathrm{rank}(K) = n$. \ 
The example above is neither positive semi-definite nor a projector. \ However, we conjecture that this example can be embedded into a larger matrix that does satisfy these constraints. 

Other directions for future work include improving the copy and time complexities of our algorithm, or giving conditional or unconditional lower bounds. \
Currently, the best lower bound we know is that $\Omega(m/\log m)$ measurements are needed, just from an information-theoretic argument (each measurement gives at most $n \log m$ bits of information and the state is characterized by $2nm$ real parameters). \ Similarly, $\Omega(mn)$ time is needed just to write down the output.

Since $\mathrm{rank}(K) = n$, it should be possible to reduce the number of measurement bases from $O(m)$ to $O(n)$ (perhaps with the low-rank matrix recovery techniques used in \cite{kueng2017low}). \ Doing so would yield an immediate improvement in the copy and time complexities of our algorithm. 

Also, just as our algorithm can be adapted to quantum state tomography, it is possible that the converse holds. \ Can quantum state tomography algorithms (in the entangled or unentangled measurement setting) be adapted to non-interacting fermion state tomography? \ Also, do quantum state tomography lower bounds imply lower bounds for learning non-interacting fermion state? \ In analogy with quantum state tomography, perhaps $\Theta(mn)$ copies are optimal to learn non-interacting fermion states with entangled measurements and $\Theta(mn^2)$ copies are optimal with unentangled measurements. 

It would be interesting to generalize our algorithm---for example, to superpositions over different numbers of fermions, or fermionic circuits that take inputs---and to find other classes of quantum states that admit efficient learning algorithms (for example, perhaps low-entanglement states or the outputs of small-depth circuits or low-stabilizer-complexity states \cite{grewal2022low}). \  We remark that \cite{PRXQuantum.3.020328} gives evidence that generalizing our algorithm to superpositions over different numbers of fermions may not be possible unless one limits the number of terms in the superposition. \
On the other hand, \cite[Theorem 8]{PRXQuantum.3.020328} might be useful in developing learning algorithms for matchgate circuits. 

Finally, what can be said about learning non-interacting \emph{boson} states? \ The goal would be to reconstruct an $m \times n$ column-orthonormal matrix $A$ given copies of a non-interacting boson state. \
However, the boson case is even trickier than the fermion case. \
In particular, boson statistics no longer depend only on the inner products between the rows of $A$, the way fermion statistics do. \
Indeed, even if we collected enough information to reconstruct a bosonic state, it seems that any algorithm would have to solve a quite complicated set of nonlinear equations. 

\section*{Acknowledgements}
We thank Andrew Zhao for notifying us that the previous version of this manuscript contained an error and providing other insightful comments. \
We also thank Yuxuan Zhang, Alex Kulesza, Ankur Moitra, William Kretschmer, Dax Enshan Koh, Andrea Rocchetto, and Patrick Rall for helpful discussions, and Alex Arkhipov, William Kretschmer, and Daniel Liang for helpful comments on a previous version of this manuscript. 

\bibliographystyle{alphaurl}
\bibliography{fermion}

\newcommand{\etalchar}[1]{$^{#1}$}
\begin{thebibliography}{WHLB22}

\bibitem[AA14]{aaronson2013bosonsampling}
Scott Aaronson and Alex Arkhipov.
\newblock Boson{S}ampling is {F}ar {F}rom {U}niform.
\newblock {\em Quantum Information and Computation}, 14(15-16):1383--1423,
  2014.

\bibitem[Aar07]{aaronson2007learnability}
Scott Aaronson.
\newblock The learnability of quantum states.
\newblock {\em Proceedings of the Royal Society A: Mathematical, Physical and
  Engineering Sciences}, 463(2088):3089--3114, 2007.
\newblock \href {https://doi.org/10.1098/rspa.2007.0113}
  {\path{doi:10.1098/rspa.2007.0113}}.

\bibitem[Aar20]{doi:10.1137/18M120275X}
Scott Aaronson.
\newblock {Shadow Tomography of Quantum States}.
\newblock {\em SIAM Journal on Computing}, 49(5):STOC18--368--STOC18--394,
  2020.
\newblock \href {https://doi.org/10.1145/3188745.3188802}
  {\path{doi:10.1145/3188745.3188802}}.

\bibitem[Aar22]{aaronson2022introduction}
Scott Aaronson.
\newblock {Introduction to Quantum Information Science II Lecture Notes}.
\newblock 2022.
\newblock
  \href{https://scottaaronson.com/qisii.pdf}{scottaaronson.com/qisii.pdf}.

\bibitem[ABDY22]{arunachalam2022phase}
Srinivasan Arunachalam, Sergey Bravyi, Arkopal Dutt, and Theodore~J. Yoder.
\newblock Optimal algorithms for learning quantum phase states, 2022.
\newblock \href {https://doi.org/10.48550/arxiv.2208.07851}
  {\path{doi:10.48550/arxiv.2208.07851}}.

\bibitem[AG21]{aaronson2021efficient}
Scott Aaronson and Sabee Grewal.
\newblock {Efficient Learning of Non-Interacting Fermion Distributions}.
\newblock {\em arXiv preprint arXiv:2102.10458v2}, 2021.

\bibitem[BO21]{10.1145/3406325.3451109}
Costin B\u{a}descu and Ryan O'Donnell.
\newblock {Improved Quantum Data Analysis}.
\newblock In {\em Proceedings of the 53rd Annual ACM SIGACT Symposium on Theory
  of Computing}, STOC 2021, page 1398–1411, 2021.
\newblock \href {https://doi.org/10.1145/3406325.3451109}
  {\path{doi:10.1145/3406325.3451109}}.

\bibitem[Can20]{canonne2020short}
Cl{\'e}ment~L. Canonne.
\newblock A short note on learning discrete distributions.
\newblock {\em arXiv preprint arXiv:2002.11457}, 2020.
\newblock \href {https://doi.org/10.48550/arXiv.2002.11457}
  {\path{doi:10.48550/arXiv.2002.11457}}.

\bibitem[CPF{\etalchar{+}}10]{cramer2010efficient}
Marcus Cramer, Martin~B. Plenio, Steven~T. Flammia, Rolando Somma, David Gross,
  Stephen~D. Bartlett, Olivier Landon-Cardinal, David Poulin, and Yi-Kai Liu.
\newblock Efficient quantum state tomography.
\newblock {\em Nature communications}, 1(1):1--7, 2010.
\newblock \href {https://doi.org/10.1038/ncomms1147}
  {\path{doi:10.1038/ncomms1147}}.

\bibitem[GIKL22]{grewal2022low}
Sabee Grewal, Vishnu Iyer, William Kretschmer, and Daniel Liang.
\newblock {Low-Stabilizer-Complexity Quantum States Are Not Pseudorandom}.
\newblock {\em arXiv preprint arXiv:2209.14530}, 2022.
\newblock \href {https://doi.org/10.48550/arXiv.2209.14530}
  {\path{doi:10.48550/arXiv.2209.14530}}.

\bibitem[GKKT20]{guctua2020fast}
Madalin Gu{\c{t}}{\u{a}}, Jonas Kahn, Richard Kueng, and Joel~A. Tropp.
\newblock Fast state tomography with optimal error bounds.
\newblock {\em Journal of Physics A: Mathematical and Theoretical},
  53(20):204001, 2020.
\newblock \href {https://doi.org/10.1088/1751-8121/ab8111}
  {\path{doi:10.1088/1751-8121/ab8111}}.

\bibitem[HHJ{\etalchar{+}}17]{haah2017sample}
Jeongwan Haah, Aram~W. Harrow, Zhengfeng Ji, Xiaodi Wu, and Nengkun Yu.
\newblock Sample-{O}ptimal {T}omography of {Q}uantum {S}tates.
\newblock {\em IEEE Transactions on Information Theory}, 63(9):5628--5641,
  2017.
\newblock \href {https://doi.org/10.1109/TIT.2017.2719044}
  {\path{doi:10.1109/TIT.2017.2719044}}.

\bibitem[HJ94]{horn1994topics}
Roger~A. Horn and Charles~R. Johnson.
\newblock {\em {Topics in Matrix Analysis}}.
\newblock Cambridge University Press, 1994.
\newblock \href {https://doi.org/10.1017/CBO9780511840371}
  {\path{doi:10.1017/CBO9780511840371}}.

\bibitem[HK64]{hohenberg1964inhomogeneous}
Pierre Hohenberg and Walter Kohn.
\newblock {Inhomogeneous Electron Gas}.
\newblock {\em Physical Review}, 136(3B):B864, 1964.
\newblock \href {https://doi.org/10.1103/PhysRev.136.B864}
  {\path{doi:10.1103/PhysRev.136.B864}}.

\bibitem[HKP20]{huang2020predicting}
Hsin-Yuan Huang, Richard Kueng, and John Preskill.
\newblock Predicting many properties of a quantum system from very few
  measurements.
\newblock {\em Nature Physics}, 16(10):1050--1057, 2020.
\newblock \href {https://doi.org/10.1038/s41567-020-0932-7}
  {\path{doi:10.1038/s41567-020-0932-7}}.

\bibitem[Kni01]{knill2001fermionic}
Emanuel Knill.
\newblock Fermionic {L}inear {O}ptics and {M}atchgates.
\newblock {\em arXiv preprint quant-ph/0108033}, 2001.
\newblock \href {https://doi.org/10.48550/arXiv.quant-ph/0108033}
  {\path{doi:10.48550/arXiv.quant-ph/0108033}}.

\bibitem[KRT17]{kueng2017low}
Richard Kueng, Holger Rauhut, and Ulrich Terstiege.
\newblock {Low Rank Matrix Recovery From Rank One Measurements}.
\newblock {\em Applied and Computational Harmonic Analysis}, 42(1):88--116,
  2017.
\newblock \href {https://doi.org/10.1016/j.acha.2015.07.007}
  {\path{doi:10.1016/j.acha.2015.07.007}}.

\bibitem[KS65]{kohn1965self}
Walter Kohn and Lu~Jeu Sham.
\newblock {Self-Consistent Equations Including Exchange and Correlation
  Effects}.
\newblock {\em Physical Review}, 140(4A):A1133, 1965.
\newblock \href {https://doi.org/10.1103/PhysRev.140.A1133}
  {\path{doi:10.1103/PhysRev.140.A1133}}.

\bibitem[Low22]{low2022classical}
Guang~Hao Low.
\newblock Classical shadows of fermions with particle number symmetry, 2022.
\newblock \href {https://doi.org/10.48550/arxiv.2208.08964}
  {\path{doi:10.48550/arxiv.2208.08964}}.

\bibitem[Mac75]{macchi1975coincidence}
Odile Macchi.
\newblock {The Coincidence Approach to Stochastic Point Processes}.
\newblock {\em Advances in Applied Probability}, 7(1):83--122, 1975.
\newblock \href {https://doi.org/10.2307/1425855} {\path{doi:10.2307/1425855}}.

\bibitem[Mon17]{montanaro2017learning}
Ashley Montanaro.
\newblock Learning stabilizer states by {B}ell sampling.
\newblock {\em arXiv:1707.04012}, 2017.
\newblock \href {https://doi.org/10.48550/arXiv.1707.04012}
  {\path{doi:10.48550/arXiv.1707.04012}}.

\bibitem[ODMZ22]{PRXQuantum.3.020328}
Micha\l{} Oszmaniec, Ninnat Dangniam, Mauro~E.S. Morales, and Zolt\'an
  Zimbor\'as.
\newblock {Fermion Sampling: A Robust Quantum Computational Advantage Scheme
  Using Fermionic Linear Optics and Magic Input States}.
\newblock {\em PRX Quantum}, 3:020328, 2022.
\newblock \href {https://doi.org/10.1103/PRXQuantum.3.020328}
  {\path{doi:10.1103/PRXQuantum.3.020328}}.

\bibitem[O'G22]{ogorman2022fermionic}
Bryan O'Gorman.
\newblock Fermionic tomography and learning, 2022.
\newblock \href {https://doi.org/10.48550/arxiv.2207.14787}
  {\path{doi:10.48550/arxiv.2207.14787}}.

\bibitem[OW16]{o2016efficient}
Ryan O'Donnell and John Wright.
\newblock Efficient {Q}uantum {T}omography.
\newblock In {\em Proceedings of the Forty-Eighth Annual ACM Symposium on
  Theory of Computing}, pages 899--912, 2016.
\newblock \href {https://doi.org/10.1145/2897518.2897544}
  {\path{doi:10.1145/2897518.2897544}}.

\bibitem[QC20]{google2020hartree}
Google~AI Quantum and Collaborators.
\newblock {Hartree-Fock on a superconducting qubit quantum computer}.
\newblock {\em Science}, 369(6507):1084--1089, 2020.
\newblock URL: \url{https://10.1126/science.abb9811}, \href
  {https://doi.org/10.1126/science.abb9811}
  {\path{doi:10.1126/science.abb9811}}.

\bibitem[RKT15]{rising2015efficient}
Justin Rising, Alex Kulesza, and Ben Taskar.
\newblock {An Efficient Algorithm for the Symmetric Principal Minor Assignment
  Problem}.
\newblock {\em Linear Algebra and its Applications}, 473:126--144, 2015.
\newblock \href {https://doi.org/10.1016/j.laa.2014.04.019}
  {\path{doi:10.1016/j.laa.2014.04.019}}.

\bibitem[RZBB94]{PhysRevLett.73.58}
Michael Reck, Anton Zeilinger, Herbert~J. Bernstein, and Philip Bertani.
\newblock Experimental realization of any discrete unitary operator.
\newblock {\em Phys. Rev. Lett.}, 73:58--61, 1994.
\newblock \href {https://doi.org/10.1103/PhysRevLett.73.58}
  {\path{doi:10.1103/PhysRevLett.73.58}}.

\bibitem[TD02]{terhal2002classical}
Barbara~M. Terhal and David~P. DiVincenzo.
\newblock Classical simulation of noninteracting-fermion quantum circuits.
\newblock {\em Physical Review A}, 65(3):032325, 2002.
\newblock \href {https://doi.org/10.1103/PhysRevA.65.032325}
  {\path{doi:10.1103/PhysRevA.65.032325}}.

\bibitem[UBMR17]{urschel2017learning}
John Urschel, Victor-Emmanuel Brunel, Ankur Moitra, and Philippe Rigollet.
\newblock {Learning Determinantal Point Processes with Moments and Cycles}.
\newblock In {\em International Conference on Machine Learning}, pages
  3511--3520. PMLR, 2017.

\bibitem[Val02]{valiant2002quantum}
Leslie~G. Valiant.
\newblock Quantum {C}ircuits {T}hat {C}an {B}e {S}imulated {C}lassically in
  {P}olynomial {T}ime.
\newblock {\em SIAM Journal on Computing}, 31(4):1229--1254, 2002.
\newblock \href {https://doi.org/10.1137/S0097539700377025}
  {\path{doi:10.1137/S0097539700377025}}.

\bibitem[Wey12]{weyl1912asymptotische}
Hermann Weyl.
\newblock {Das asymptotische Verteilungsgesetz der Eigenwerte linearer
  partieller Differentialgleichungen (mit einer Anwendung auf die Theorie der
  Hohlraumstrahlung)}.
\newblock {\em Mathematische Annalen}, 71(4):441--479, 1912.
\newblock \href {https://doi.org/10.1007/BF01456804}
  {\path{doi:10.1007/BF01456804}}.

\bibitem[WHLB22]{wan2022matchgate}
Kianna Wan, William~J. Huggins, Joonho Lee, and Ryan Babbush.
\newblock {Matchgate Shadows for Fermionic Quantum Simulation}, 2022.
\newblock \href {https://doi.org/10.48550/arxiv.2207.13723}
  {\path{doi:10.48550/arxiv.2207.13723}}.

\bibitem[Wic50]{wick1950evaluation}
Gian-Carlo Wick.
\newblock {The Evaluation of the Collision Matrix}.
\newblock {\em Physical Review}, 80(2):268, 1950.
\newblock \href {https://doi.org/10.1103/PhysRev.80.268}
  {\path{doi:10.1103/PhysRev.80.268}}.

\bibitem[ZRM21]{zhao2021fermionic}
Andrew Zhao, Nicholas~C. Rubin, and Akimasa Miyake.
\newblock {Fermionic Partial Tomography via Classical Shadows}.
\newblock {\em Phys. Rev. Lett.}, 127:110504, 2021.
\newblock \href {https://doi.org/10.1103/PhysRevLett.127.110504}
  {\path{doi:10.1103/PhysRevLett.127.110504}}.

\end{thebibliography}

\end{document}